%% file: main.tex
\documentclass{theoretics}

\title{Conditional Dichotomy of Boolean Ordered Promise CSPs}

\ThCSauthor[affil1]{Joshua Brakensiek}{jbrakens@cs.stanford.edu}[https://orcid.org/0000-0003-4149-72988]
\ThCSauthor[affil2]{Venkatesan Guruswami}{venkatg@berkeley.edu}[https://orcid.org/0000-0001-7926-3396]
\ThCSauthor[affil2]{Sai Sandeep}{saisandeep@berkeley.edu}[https://orcid.org/0000-0003-2538-4250]

\ThCSaffil[affil1]{Department of Computer Science, Stanford University, Stanford, CA, USA}
\ThCSaffil[affil2]{Electrical Engineering and Computer Sciences Department, University of California, Berkeley}

\ThCSyear{2023}
\ThCSarticlenum{2}
\ThCSreceived{Jan 19, 2022}
\ThCSrevised{Dec 8, 2022}
\ThCSaccepted{Dec 16, 2022}
\ThCSpublished{Jan 24, 2023}
\ThCSkeywords{constraint satisfaction problems, Rich 2-to-1 conjecture, Shapley value, polymorphisms}
\ThCSdoi{10.46298/theoretics.23.2}
\ThCSshortnames{J. Brakensiek, V. Guruswami and S. Sandeep}
\ThCSshorttitle{Conditional Dichotomy of Boolean Ordered Promise CSPs}
\ThCSthanks{A previous version of this paper appeared in ICALP 2021. Joshua Brakensiek was supported in part by an NSF Graduate Research Fellowship and a Microsoft Research PhD Fellowship. Venkatesan Guruswami was supported in part by NSF grants CCF-2228287 and CCF-2211972 and a Simons Investigator award. Sai Sandeep was supported in part by NSF grant CCF-2228287.}

\usepackage{amsmath}
\usepackage{cleveref}
\usepackage{tikz}

\newcommand{\B}{\mathcal{B}}
\newcommand{\E}{\mathbb{E}}
\newcommand{\bits}{\{0,1\}}

\newcommand{\si}{\Phi}

\begin{document}
\maketitle

\begin{abstract}
    Promise Constraint Satisfaction Problems (PCSPs) are a generalization of Constraint Satisfaction Problems (CSPs) where each predicate has a strong and a weak form and given a CSP instance, the objective is to distinguish if the strong form can be satisfied vs.\,even the weak form cannot be satisfied. Since their formal introduction by Austrin, Guruswami, and H\aa stad~\cite{AGH17}, there has been a flurry of works on PCSPs, including recent breakthroughs in approximate graph coloring~\cite{barto2019algebraic,KO19,WZ19}. The key tool in studying PCSPs is the algebraic framework developed in the context of CSPs where the closure properties of the satisfying solutions known as \emph{polymorphisms} are analyzed. 
    
    The polymorphisms of PCSPs are significantly richer than CSPs---this is illustrated by the fact that even in the Boolean case, we still do not know if there exists a dichotomy result for PCSPs analogous to Schaefer's dichotomy result~\cite{Schaefer78} for CSPs. In this paper, we study a special case of Boolean PCSPs, namely Boolean \textit{Ordered} PCSPs where the Boolean PCSPs have the predicate $x \leq y$. In the algebraic framework, this is the special case of Boolean PCSPs when the polymorphisms are \emph{monotone functions}. 
    We prove that Boolean Ordered PCSPs exhibit a computational dichotomy assuming the Rich $2$-to-$1$ Conjecture~\cite{BravermanKM19} which is a perfect completeness surrogate of the Unique Games Conjecture. 
    
    In particular, assuming the Rich $2$-to-$1$ Conjecture, we prove that a Boolean Ordered PCSP can be solved in polynomial time if for every $\epsilon >0$, it has polymorphisms where each coordinate has \emph{Shapley value} at most $\epsilon$, else it is NP-hard. 
    The algorithmic part of our dichotomy result is based on a result that if a monotone Boolean function has all Shapley values small, then it has a large threshold function as a minor.
    For the conditional hardness result, we show that Shapley value behaves in a consistent manner under a uniformly random $2$-to-$1$ minor. As a structural result of independent interest, we construct an example to show that the Shapley value can behave inconsistently with respect to an arbitrary $2$-to-$1$ minor. 
\end{abstract}

\section{Introduction}
\label{sec:intro}

Constraint satisfaction problems (CSP) have played a very influential role in the theory of computation, providing an excellent testbed for the development of both algorithmic and hardness techniques, which then extend to more general settings. A CSP over domain $D$ is specified by a finite collection $\mathbb{A}$ of predicates over $D$, and is denoted as CSP($\mathbb{A}$). Given an input containing $n$ variables with constraints on the variables using these predicates, the objective is to identify if we can assign values from $D$ to the variables that satisfies all the constraints. Examples of CSPs include classical problems such as $3$-SAT and $3$-Coloring of graphs. 

When the domain is Boolean, Schaefer~\cite{Schaefer78} proved that every CSP is either in P or is NP-Complete. Feder and Vardi~\cite{FederV98} conjectured that the same should hold over arbitrary domains as well. They also showed that the then known algorithmic results all follow by the algebraic closure properties of the CSPs. 
This notion was formalized by Jeavons, Cohen, and Gyssens~\cite{JeavonsCG97, Jeavons98} and other works~\cite{BulatovJK05} that crystallized the \textit{(universal) algebraic approach} to CSPs.
In the algebraic approach, the higher-order closure properties obeyed by the predicates, namely their \emph{polymorphisms}, are studied. 
A polymorphism is a function that, when applied coordinate-wise to arbitrary satisfying assignments to the predicate, is guaranteed to produce an output that satisfies the predicate.
For example, consider an arbitrary instance $I$ of the $2$-SAT problem over $n$ variables, and suppose that $\textbf{x}, \textbf{y},\textbf{z} \in \{0,1\}^n$ are three assignments that satisfy all the constraints in $I$. Now, if we compute $\textbf{u} \in \{0,1\}^n$ that is obtained by setting $u_i =\textsf{MAJ}(x_i, y_i,z_i)$ for all $i \in [n]$, the assignment $\textbf{u}$ also satisfies all the constraints of $I$. Thus, the majority function on $3$ bits is a polymorphism of the $2$-SAT CSP. 
On the other hand, for the $3$-SAT problem, it is not hard to prove that the only polymorphisms are the dictator functions. 
The algebraic approach has been immensely successful and culminated in the recent resolution of Feder-Vardi conjecture by Bulatov~\cite{Bulatov17} and Zhuk~\cite{Zhuk20}. Further, these proofs yield a precise understanding of the mathematical
structure underlying efficient algorithms: if the CSP has a ``non-trivial''
polymorphisms, the CSP is polytime solvable, and otherwise, it is NP-complete.

In this paper, we study Promise Constraint Satisfaction Problems (PCSPs) that vastly generalize the CSPs. In the PCSPs, each predicate has a weak and a strong form--given an instance of PCSP containing $n$ variables with the constraints, the goal is to distinguish between the case that the stronger form can be satisfied vs.\,even the weaker one cannot be satisfied. A classical example of a PCSP is the approximate graph coloring problem, where given a graph $G$, the goal is to distinguish between the cases that $G$ can be colored with $c$ colors vs.\,it cannot be colored with $s$ colors for some $c \leq s$. 
Another example is the ($1$-in-$3$ SAT, NAE-$3$-SAT), wherein given a $1$-in-$3$-SAT instance that is promised to be satisfiable, the objective is to assign $0,1$ values to the variables such that each constraint is satisfied as in a NAE-$3$-SAT instance, i.e., both $0$ and~$1$ occur in every constraint. 
While the individual CSPs, namely $1$-in-$3$-SAT and NAE-$3$-SAT are both NP-hard, the above PCSP is in~P. 
The study of PCSPs was formally initiated by Austrin, Guruswami, and H\aa stad~\cite{AGH17}, and since then, there has been a lot of recent interest in PCSPs, including the development of a systematic theory in \cite{BG18,barto2019algebraic} and leading to breakthroughs in approximate graph coloring~\cite{barto2019algebraic, KO19, WZ19}.

The central question in the study of PCSPs is whether there exists a complexity dichotomy for PCSPs, i.e., if every PCSP is either in P or is NP-complete. 
As is the case with CSPs, the key tool towards establishing a potential dichotomy result is the algebraic approach.
The Galois correspondence from the CSP world extends to PCSPs, i.e., the polymorphisms fully capture the computational complexity of the underlying PCSP~\cite{P02, BG18}.
This has been extended to show that just the identities satisfied by the polymorphisms suffice to capture the computational complexity of the underlying PCSP~\cite{barto2019algebraic}. 
However, the polymorphisms of PCSPs are much richer, and characterizing which polymorphisms lead to algorithms and which ones lead to hardness has been a challenging problem. Conceptually, the principal difficulty is that the polymorphisms for CSPs are closed under composition (hence referred to as \textit{clones}), whereas for PCSPs, this is no longer the case.

As a result, even in the Boolean case, we do not have a dichotomy theorem for PCSPs.
Towards establishing a potential Boolean PCSP dichotomy, progress has been made by Ficak, Kozik, Ols\'{a}k and Stankiewicz~\cite{FicakKOS19}, who obtained a dichotomy result when each predicate is symmetric. In this paper, we study Boolean PCSPs that contain the \emph{simplest non-symmetric predicate}, $x \rightarrow y$. We call such Boolean PCSPs \textit{Ordered} as we can also view the implication constraint as an ordering requirement $x \leq y$\footnote{As PCSPs have pairs of predicates, the ordering predicate pair has both the strong and weak forms as $x \leq y$, i.e., $\{(0,0),(0,1),(1,1)\}$}.

Ordered Boolean PCSPs have come under recent study. The work of Petr~\cite{petr20} (inspired by work of Barto~\cite{Bar18,Bar20}) considered a special class of Ordered Boolean PCSPs which have an additional predicate $x \neq y$ (this corresponds to allowing negations in the constraints) as well as the requirement that the majority on three bits is \emph{not} a polymorphism. In this setting Petr was able to show that such Ordered Boolean PCSPs are NP-hard. However, the approach considered does not seem immediately extendable to analyzing general Ordered Boolean PCSPs~\cite{Bar20}.

The main motivation for studying these PCSPs comes from the fact that adding the additional $ x \leq y$ predicate is equivalent to restricting the polymorphisms of the PCSPs to be \emph{monotone functions}. Monotonicity is an influential theme in the study of Boolean functions and complexity theory, and understanding the structure of polymorphisms in the monotone case is an important (and certainly necessary) subcase towards a general characterization of polymorphisms vs.\,tractability for arbitrary Boolean PCSPs. For the special case of Boolean Ordered PCSPs which include negation constraints, it was conjectured in~\cite{Bar20} that polynomial time tractability is characterized by the existence of majority polymorphisms of arbitrarily large arity.

Our main result is that Boolean Ordered PCSPs exhibit a dichotomy, under the recently introduced \emph{Rich $2$-to-$1$ Conjecture} of Braverman, Khot, and Minzer~\cite{BravermanKM19}.   

\begin{theorem}
\label{thm:main}
Assuming the Rich $2$-to-$1$ Conjecture, every Ordered Boolean PCSP is either in P or is NP-Complete. In particular, an Ordered PCSP $\Gamma$ is in P if for every $\epsilon >0$, there are polymorphisms of\, $\Gamma$ with every coordinate having Shapley value at most $\epsilon$, else it is NP-Complete. Equivalently, $\Gamma$ is in P if it has threshold\footnote{We call a Boolean function $f:\{0,1\}^n \to \{0,1\}$ a threshold function if there is an integer $t$ such that for every $\textbf{x} \in \{0,1\}^n$, $f(\textbf{x})=1$ if and only if $|\{i \in [n]: x_i = 1\}|\geq t$.}
polymorphisms of arbitrarily large arity, else it is NP-Complete.
\end{theorem}

As a concrete example, recall the earlier mentioned example of ($1$-in-$3$-SAT, NAE-$3$-SAT). As it has threshold polymorphisms of arbitrarily large arity, it remains polynomial time solvable even after adding the predicate $x \rightarrow y$. However, if we also add another two-variable predicate $x \neq y$, the PCSP no longer has threshold polymorphisms, and by our above result, it becomes NP-Complete. 
 
We obtain the conditional dichotomy result by analyzing the polymorphisms of the Ordered PCSPs. The key idea in the algebraic approach to PCSPs is that the PCSP is tractable if the polymorphisms are close to symmetric, and the PCSP is hard if all the polymorphisms have a small number of ``important'' coordinates. More concretely, on the algorithmic front, it has been proved that symmetric polymorphisms of arbitrarily large arities lead to polynomial time algorithms for PCSPs~\cite{BrakensiekGWZ20}. 
On the hardness side, if all the polymorphisms depend on a bounded number of coordinates, then the underlying PCSP is NP-hard~\cite{AGH17}. This has been extended to various other notions, including combinatorial ones such as $C$-fixing~\cite{BG16}, and topological ones such as having a bounded number of coordinates with non-zero winding number~\cite{KO19}. In this paper, we study the monotone polymorphisms using analytical techniques.

In particular, we use Shapley value to analyze the monotone polymorphisms. 
For a monotone function $f:\{0,1\}^n \rightarrow \{0,1\}$, the Shapley value of a coordinate $i$ is the probability that on a random path from $\{0,0,\ldots, 0\}$ to $\{1,1,\ldots, 1\}$, the function value turns from $0$ to $1$ when we switch the $i$th coordinate to $1$. 
Initially studied to understand the power of an individual in voting systems~\cite{ShapleyS54}, Shapley value has now found applications in various settings, especially in game theory~\cite{MichalakASRJ13, NarayanamN11}.  
In our setting, there are two advantages of using Shapley value to study the polymorphisms. First, it is a relative measure of the importance of a coordinate, as opposed to other notions of Influence which are absolute. This helps in bounding the number of coordinates with Shapley value above a certain threshold. Second, it is a versatile measure with combinatorial and analytical interpretations~\cite{DeDS17} which helps in proving that Shapley value stays consistent under function minors\footnote{A minor(formally defined in~\Cref{sec:prelims}) of a function $f:\{0,1\}^m\rightarrow\{0,1\}$ is a function $g:\{0,1\}^n \rightarrow\{0,1\}$ of smaller arity $n\leq m$ obtained from $f$ by identifying sets of variables together.}, a key property necessary in both the algorithm and the hardness. 

\paragraph{Algorithm Overview.}
We obtain our algorithmic result by using the Basic Linear Programming with Affine relaxation (BLP+Affine relaxation), combined with a structural result regarding the monotone functions with bounded Shapley value. As mentioned earlier, PCSPs with symmetric polymorphisms of arbitrarily large arities can be solved in polynomial time using the BLP+Affine relaxation algorithm~\cite{BrakensiekGWZ20}. Our main structural result is that Boolean functions with bounded Shapley value have arbitrarily large threshold functions as minors. Since the set of polymorphisms of a PCSP are closed under taking minors, this proves that the underlying PCSP~$\Gamma$ has arbitrarily large threshold functions as polymorphisms, which then implies that~$\Gamma$ is in~P. The key tool underlying our structural result is a result of Kalai~\cite{Kalai04} that states that under certain conditions, monotone Boolean functions with arbitrarily small Shapley value have a sharp threshold. 

\paragraph{Hardness Overview.}
We obtain our hardness result assuming the Rich $2$-to-$1$ Conjecture. Braverman, Khot, and Minzer~\cite{BravermanKM19} introduced the conjecture as a perfect completeness surrogate of the well known Unique Games Conjecture~\cite{Khot02a}. They also proved that the conjecture is equivalent to Unique Games Conjecture when we relax the perfect completeness requirement.   
The reduction from the Rich $2$-to-$1$ Conjecture to PCSPs follows using the standard Label Cover-Long Code paradigm. 
The key ingredient in this reduction is a decoding of the Long Codes to a bounded number of coordinates that is consistent under function minors. 
We decode each Long Code function to the coordinates with $\Omega(1)$ Shapley value---as the sum of Shapley values of all the coordinates of any monotone function is equal to $1$, there is a bounded number of such coordinates. 
We argue about the consistency of this decoding using a structural result that states that under a uniformly random minor, Shapley value is roughly preserved. 

\paragraph{On the necessity of ``richness'' in  $2$-to-$1$ Conjecture.}
A natural question is whether our hardness result can be obtained using a weaker assumption such as the $2$-to-$1$ conjecture (whose imperfect completeness version was recently established~\cite{KMS17,DKKMS18a,DKKMS18,KMS18}). We shed some light on this question by showing that there are monotone Boolean functions $f:\{0,1\}^{2n}\rightarrow \{0,1\}$ and $g:\{0,1\}^n\rightarrow \{0,1\}$ such that $g$ is a minor of $f$ with respect to the $2$-to-$1$ function $\pi$, both the functions $f$ and $g$ have exactly one coordinate $i_1$, $i_2$ respectively, with $\Omega(1)$ Shapley value, and yet $\pi(i_1) \neq i_2$. 
Such an adversarial example is interesting from two angles: first, it shows that even using the $2$-to-$1$ conjecture, the Shapley value based decoding is not consistent. 
Second, it gives an example of agents pairing up maliciously to completely alter the Shapley value. 
The underlying phenomenon is that the rich $2$-to-$1$ games have ``subcode-covering'' property, which is absent in the standard $2$-to-$1$ games, helping in preserving the consistency of any biased influence measure such as the Shapley value.

\paragraph{Organization.}
In~\Cref{sec:prelims}, we formally define PCSPs, polymorphisms, and Shapley value. We present the algorithmic and hardness parts of our dichotomy result in~\Cref{sec:alg} and~\Cref{sec:hardness} respectively. We present the adversarial example of a $2$-to-$1$ minor that alters the Shapley value in~\Cref{sec:adversarial}.

\section{Preliminaries}
\label{sec:prelims}

\noindent \textbf{Notations.} We use $[n]$ to denote the set $\{1,2,\ldots, n\}$. For a $k$-ary relation $A \subseteq [q]^k$, we abuse the notation and use $A$ both as a subset of $[q]^k$, and also as a predicate $A:[q]^k \rightarrow \{0,1\}$. Similarly, for a function $f:\{0,1\}^n \to D$ and a set $S \subseteq [n]$, we sometimes use $f(S)$ to denote $f(\textbf{v})$ where $v_i = 1$ if $i \in S$, and $0$ otherwise. 
For a vector $\textbf{x}=(x_1, x_2, \ldots, x_n) \in \{0,1\}^n$, we use $\textsf{hw}(\textbf{x})$ to denote $\sum_{i=1}^n x_i$.
For two vectors $\textbf{x}, \textbf{y} \in \{0,1\}^n$, we say that $\textbf{x} \leq \textbf{y}$ if $x_i \leq y_i$ for all $i \in [n]$.
A Boolean function $f:\{0,1\}^n \rightarrow \{0,1\}$ is called monotone if $f(\textbf{x})\leq f(\textbf{y})$ for all $\textbf{x}\leq \textbf{y}$.

\paragraph{PCSPs and Polymorphisms.}
We first define Constraint Satisfaction Problems(CSP). 
\begin{definition}(CSP) Given a $k$-ary relation $A:D^k\rightarrow \{0,1\}$ over a domain $D$, the Constraint Satisfaction Problem(CSP) associated with the predicate $A$ takes a set of variables $V =\{v_1, v_2, \ldots, v_n\}$ as input which are to be assigned values from $D$. There are $m$ constraints $(e_1, e_2, \ldots, e_m)$ each consisting of $e_i = ( (e_i)_1, (e_i)_2, \ldots, (e_i)_k ) \subseteq V^k $ that indicate that the corresponding assignment should belong to $A$. 
The objective is to identify if there is an assignment $V \rightarrow D$ that satisfies all the constraints. 
\end{definition}
In general, we can have multiple relations $A_1,A_2,\dots,A_l$, and different constraints can use different relations. We denote such a CSP by $CSP(A_1,A_2,\ldots,A_l)$.  


We formally define Promise Constraint Satisfaction Problems (PCSP). 
\begin{definition}(PCSP)
	In a Promise Constraint Satisfaction Problem $PCSP(\Gamma)$ over a pair of domains $D_1, D_2$, we have a set of pairs of relations $\Gamma = \{ (A_1,B_1), (A_2,B_2),\ldots, (A_l,B_l)\}$ such that for every $i \in [l]$, $A_i$ is a subset of $D_1^{k_i}$ and $B_i$ is a subset of $D_2^{k_i}$. 
	Furthermore, there is a homomorphism $h:D_1\rightarrow D_2$ such that for all $i\in[l]$ and $x \in D_1^{k_i}$, $x \in A_i$ implies $h(x) \in B_i$.
	Given a $CSP(A_1,A_2,\ldots,A_l)$ instance, the objective is to distinguish between the two cases:
	\begin{enumerate}
		\item There is an assignment to the variables from $D_1$ that satisfies every constraint when viewed as $CSP(A_1,A_2,\ldots,A_l)$. 
		\item There is no assignment to the variables from $D_2$ that satisfies every constraint when viewed as $CSP(B_1,B_2,\ldots,B_l)$. 
	\end{enumerate}
\end{definition}

We now define Boolean Ordered PCSPs. 
\begin{definition}(Boolean Ordered PCSP)
A PCSP over a pair of domains $D_1,D_2$ with the set of pairs of relations $\Gamma = \{ (A_1,B_1), (A_2,B_2),\ldots, (A_l,B_l)\}$ is said to be Boolean Ordered if the following hold.
\begin{enumerate}
    \item The domains are both Boolean i.e., $D_1 = D_2 = \{ 0,1\}$. 
    \item There exists $i \in [l]$ such that $A_i = B_i = \{(0,0), (0,1),(1,1)\}$. 
\end{enumerate}
\end{definition}

Associated with every PCSP, there are polymorphisms that capture the closure properties of the satisfying solutions to the PCSP. More formally, we can define polymorphisms of a PCSP as follows. 
\begin{definition}(Polymorphisms)
	For $PCSP(\Gamma)$ with $\Gamma = \{ ((A_1,B_1), (A_2, B_2), \ldots, (A_l,B_l))\}$ where for every $i \in [l]$,
	$A_i:[q_1]^{k_i} \rightarrow \{0,1\}, B_i:[q_2]^{k_i} \rightarrow \{0,1\}$, a polymorphism of arity $n$ is a function $f:[q_1]^n \rightarrow [q_2]$ that satisfies the below property for all $i \in [l]$.
	For all $(\textbf{v}_1,\textbf{v}_2,\ldots,\textbf{v}_{k_i})$ such that for all $j \in [n], ( (\textbf{v}_1)_j, (\textbf{v}_2)_j, \ldots, (\textbf{v}_{k_i})_j ) \in A_i$, we have 
	\[
	(f(\textbf{v}_1), f(\textbf{v}_2), \ldots , f(\textbf{v}_{k_i})) \in B_i
	\]
	We use $\emph{Pol}(\Gamma)$ to denote the family of all the polymorphisms of $PCSP(\Gamma)$.
\end{definition}
A crucial property satisfied by $\text{Pol}(\Gamma)$ is that the family of functions is closed under taking minors. We first define the minor of a function formally. 
\begin{definition}(Minor of a function)
 For a Boolean function $f: [q]^n \rightarrow [q']$ and an integer $m$,  the function $g : [q]^m \rightarrow [q']$ is said to be a minor of $f$ with respect to the function $\pi : [n] \rightarrow [m]$ if 
 \[
 g(x_1, x_2, \ldots, x_m) = f( x_{\pi(1)}, x_{\pi(2)}, \ldots, x_{\pi(n)} )\quad \forall x_1, x_2, \ldots, x_m \in [q].
 \]
 We say that a function $g$ is a minor of $f$ if there exists some $\pi$ such that $g$ is a minor of $f$ with respect to $\pi$. 
\end{definition}

We are often interested in $2$-to-$1$ minors. A function $g$ is said to be a $2$-to-$1$ minor of $f$ if there exists a $2$-to-$1$ function $\pi$ such that $g$ is a minor of $f$ with respect to $\pi$, where $2$-to-$1$ function is defined below. 
\begin{definition}($2$-to-$1$ function)
A function $\pi : [2n] \rightarrow [n]$ is said to be a $2$-to-$1$ function if 
\[
|\pi^{-1}(i)|=2 \,\, \forall i \in [n]
\]
\end{definition}

By the definition of the polymorphisms, we can infer that if $f \in \text{Pol}(\Gamma)$ for a PCSP $\Gamma$, then for all functions $g$ such that $g$ is a minor of $f$, we have $g \in \text{Pol}(\Gamma)$. Such a family of functions that is closed under taking minors is called as a \textit{minion}. We often refer to the family of polymorphisms of a PCSP as the polymorphism minion. 

We refer the reader to~\cite{barto2019algebraic} for an extensive introduction to PCSPs and polymorphisms.

\paragraph{Shapley value.}
Let $f:\{0,1\}^{n}\rightarrow \{0,1\}$ be a monotone Boolean function. 
We can view the monotone Boolean function $f$ as a voting scheme between two parties, and $n$ agents: the winner of the voting scheme when the $i$th agent votes for $\textbf{x}_i \in \{0,1\}$ is $f(\textbf{x})$.
The relative power of an agent in a voting scheme is typically measured using the Shapley-Shubix Index, also known as Shapley Value. 

Informally speaking, the Shapley Value of a coordinate $i$ is the probability that the $i$th agent is the altering vote when we start with all zeroes and flip the votes in a uniformly random order. More formally, 
\begin{definition}(Shapley value)
Let $f: \{0,1\}^n \rightarrow \{0,1\}$ be a monotone Boolean function. Let $\sigma \in S_n$ be a uniformly random permutation of $[n]$. For an integer $j \in [n]$, let $P_j$ denote the the set of first $j$ elements of $\sigma$ i.e., $P_j := \{ \sigma(1), \sigma(2), \ldots,\sigma(j)\}$. 
The Shapley value  $\si_f(i)$ of the coordinate $i \in [n]$ is defined as 
\[
\Phi_f(i) := \text{Pr}_{\sigma} \left\{\exists j \in [n]: \sigma(j)=i, f( P_{j-1})=0,f( P_j )=1  \right\} 
\]
\end{definition}

We also give an alternate definition of Shapley value using the notion of boundary of a coordinate.
For a monotone Boolean function $f:\{0,1\}^n \rightarrow \{0,1\}$ and coordinate $i \in [n]$, let $\mathcal{B}_f(i)$ denote the boundary of the coordinate $i$ i.e.,
\[
\B_f(i):=\{S \subseteq [n]\setminus\{i\} : f(\{i\} \cup S)=1, f(S)=0\}
\]
By the monotonicity of $f$, we can infer that $\B_f(i)$ satisfies the following sandwich property that will be useful later. 
\begin{proposition}
\label{prop:sandwich}
Let $f: \{ 0,1\}^n \rightarrow \{0,1\}$ be a monotone Boolean function and let $i \in [n]$. Then, for every pair of sets $S_1, S_2 \in \B_f(i)$ with $S_1 \subseteq S_2$, we have $S \in \B_f(i)$ for all $S$ such that $S_1 \subseteq S \subseteq S_2$.
\end{proposition}
\begin{proof}
By the monotonicity of $f$, we have $f(S \cup \{i\})\geq f(S_1 \cup \{i\})=1$, and thus, $f(S \cup \{i\})=1$. Similarly, we have $f(S)\leq f(S_2)=0$, and thus, $f(S)=0$.
\end{proof}
For an index $j \in \{0,1,\ldots,n-1\}$, 
let $\mu_f(j)^{(i)}$ denote the fraction of subsets of $[n]$ of size $j$ that are in $\B_f(i)$ i.e., 
\[
\mu_f(j)^{(i)}:=\left|\B_f(i) \cap \tbinom{[n]}{j}\right| \Big/ \tbinom{n}{j}.
\]
We can rewrite the definition of Shapley value of the $i$th coordinate as the following~\cite{Weber77}:
\begin{equation}
\label{eq:si-alternate}
\si_f(i)=\frac{\sum_{j=0}^{n-1} \mu_f(j)^{(i)}}{n}  \ .
\end{equation}

\section{Algorithm when Shapley values are small}
\label{sec:alg}

In this section, we show that monotone Boolean functions where each coordinate has bounded Shapley value has arbitrarily large threshold functions as minors, thereby proving the algorithmic part of our dichotomy result. 

Let $L$ be a positive integer and $0 \leq \tau\leq L+1$ be a non-negative integer. We let $\textsf{THR}_{L,\tau} : \{ 0, 1 \} ^ L \rightarrow \{ 0,1\}$ be the threshold function on $L$ variables with threshold $\tau$. More formally, 
\[
\textsf{THR}_{L,\tau}(\textbf{x}) := 
\begin{cases}
1 \text{ if }\textsf{hw}(\textbf{x})\geq \tau \\ 
0 \text { otherwise.}
\end{cases}
\]

For a monotone Boolean function $f:\{0,1\}^n \rightarrow \{0,1\}$ and real number $p\in [0,1]$, let $P_{p}(f)$ denote the expected value of $f(x)$ where each element $x_i, i \in [n]$ is independently set to be $1$ with probability $p$ and $0$ with probability $1-p$. For every monotone function $f$, the function $P_p(f)$ is a strictly monotone continuous function in $p$ on the interval $[0,1]$. The value $p_c=p_c(f)$ at which $P_{p_c}(f)=\frac 12$ is called the \textit{critical probability} of $f$. 

Using the Russo-Margulis Lemma~\cite{Russo1982,Margulis74} and Poincar\'{e} Inequality, we can show the following lemma that we need later.  
\begin{lemma}[Exercise 8.29(e) in~\cite{analysisOdonnell}]
\label{lem:russo}
Let $f$ be a non-constant monotone Boolean function with critical probability $p_c \le \frac 12$. Let $p_1 := \frac{1}{(2\nu)^2}p_c$ for $\nu > 0$. If $p_1 \leq \frac 12$, then $P_{p_1}(f)\geq 1-\nu$.
\end{lemma}

We now define the threshold interval of $f$. 
\begin{definition}
For a monotone function $f$ and $0<\epsilon<\frac{1}{2}$, we define $T_{\epsilon}(f) := p_2 - p_1,$
where~$p_2$ and~$p_1$ are such that $P_{p_1}(f)=\epsilon, P_{p_2}(f)=1-\epsilon$. 
\end{definition}
Kalai~\cite{Kalai04} proved the following result regarding monotone Boolean functions.
\begin{theorem}
\label{thm:kalai-threshold}
For every $a,\epsilon, \gamma >0$, there exists $\delta := \delta(a,\epsilon,\gamma)>0$ such that for every monotone Boolean function $f:\{0,1\}^n \rightarrow \{0,1\}$ with $\si_f(i)\leq \delta$ for all $i \in [n]$ and $a \leq p_c(f) \leq 1-a$, then $T_{\epsilon}(f)\leq \gamma$. 
\end{theorem}
We will use this result to show that for every monotone function where each coordinate has bounded Shapley value has arbitrarily large threshold functions as minor. 

\begin{lemma}
\label{lem:threshold-minor}
For every $L\geq 2$, there exists a $\delta := \delta(L) >0$ such that the following holds. 
For any monotone Boolean function $f:\{0,1\}^n\rightarrow \{0,1\}$ with 
\[
\si_f(i) \leq \delta \ \, \forall i \in [n]
\]
there exists a positive integer $L' \in \{ L, L+1\}$ and a non-negative integer $\tau$ such that $\emph{\textsf{THR}}_{L',\tau}$ is a minor of~$f$. 
\end{lemma}
\begin{proof}
We prove the lemma by taking a uniformly random minor of $f$.

We first obtain $\delta := \delta(L)>0$ from~\Cref{thm:kalai-threshold} by setting $\epsilon = \frac{1}{2^{L+1}}, \gamma = a = \frac{1}{L^3}$.
Our goal is to show that for this parameter $\delta$, for every monotone Boolean function $f$ with each coordinate having Shapley value at most $\delta$, there exists $L' \in \{ L, L+1\}$ and $\tau$ such that $\textsf{THR}_{L',\tau}$ is a minor of~$f$. 

We assume that $f$ is a non-constant function, else we have a trivial minor by setting $\tau = 0 $ or $\tau = L'$.
Let $p_c$ be the critical probability of $f$. 

\subparagraph{Case 1: $p_c < a=\frac{1}{L^3}$.}

Let $p_1 = L^2p_c < \frac{1}{L}$. Using~\Cref{lem:russo}, we can conclude that $P_{p_1}(f)\geq 1-\frac{1}{2L}$. As $P_p(f)$ is monotone, we get that $P_{\frac 1L}(f)> 1 - \frac{1}{2L}$.
We let $g: \{ 0,1\}^L \rightarrow \{ 0,1\}$ be a uniformly random minor of $f$ i.e., we choose the function $\pi : [n] \rightarrow [L]$ by choosing each value $\pi(i)$ uniformly and independently at random from $[L]$, and we let $g$ to be the minor of $f$ with respect to $\pi$.

Note that for every $i \in [L]$, the distribution of $g(\{i\})$ over the  random minor $g$ is the same as sampling a random input to $f$ where we set each bit to $1$ with probability $\frac{1}{L}$. As $P_{\frac 1L}(f) \geq 1-\frac{1}{2L}$, we get that for each $i \in [L]$, $g(\{i\})=1$ with probability at least $1-\frac{1}{2L}$. By union bound, with probability at least $\frac{1}{2}$, $g(\{i\})=1$ for all $i \in [L]$. As $f(0,0,\ldots,0)=0$, $g(\phi)=0$ as well. Thus, with probability at least $\frac 12$, $g=\textsf{THR}_{L,1}$. Hence, $\textsf{THR}_{L,1}$ is a minor of~$f$. 

\subparagraph{Case 2: $p_c > 1-a=1-\frac{1}{L^3}$.}

Let $f^{\dagger}$ be the Boolean dual of $f$ defined as $f^\dagger (x) = 1-f(\overline{x})$. Note that $P_p(f^\dagger)=1-P_{1-p}(f)$ for all $p \in [0,1]$. Thus, $p_c(f^\dagger)=1-p_c < a$. Using the previous case, we can infer that $\textsf{THR}_{L,1}$ is a minor of $f^\dagger$ with respect to a funtion $\pi : [n] \rightarrow [L]$. The same function $\pi $ proves that $\textsf{THR}_{L,1}^{\dagger} = \textsf{THR}_{L,L}$ is a minor of $f$. 

\subparagraph{Case 3: $a \leq p_c \leq 1 - a $.}

Using~\Cref{thm:kalai-threshold}, we obtain $p_1$ such that $P_{p_1}(f)\leq \epsilon$, and $P_{p_1 + \gamma} \geq 1 -\epsilon$, where $\epsilon = \frac{1}{2^{L+1}}, \gamma = \frac{1}{L^3}$. 
As $\gamma < \frac{1}{L(L+1)}$, there exists $L' \in \{ L, L+1\}$ and $\tau \in [L']$ such that $p_1 + \gamma < \frac{\tau}{L'}$ and $p_1 > \frac{\tau -1}{L'}$. Thus, we get that $P_{\frac{\tau}{L'}}(f) > 1-\epsilon$ and $P_{\frac{\tau -1}{L'}}<\epsilon$. 
Let $g : \{ 0, 1 \}^{L'} \rightarrow \{ 0,1\}$ be a uniformly random minor of $f$ i.e., we choose $\pi : [n] \rightarrow [L']$ by setting each value uniformly and independently at random from $[L']$ and set $g$ to be the minor of $f$ with respect to $\pi$. 
For a vector $\textbf{x} \in \{ 0, 1\}^{L'}$ with $\textsf{hw}(\textbf{x})=\tau$, with probability greater than $1-\frac{1}{2^{L+1}}$, $g(\textbf{x})=1$. Similarly, for $\textbf{x} \in \{0,1\}^{L'}$ with $\textsf{hw}(\textbf{x})=\tau -1$, with probability greater than $1-\frac{1}{2^{L+1}}$, $g(\textbf{x})=0$. Thus, with non-zero probability, $g(\textbf{x})=1$ for all $x \in \{ 0,1\}^{L'}$ with $\textsf{hw}(\textbf{x})=\tau$ and $g(\textbf{x})=0$ for all $\textbf{x} \in \{0,1\}^{L'}$ with $\textsf{hw}(\textbf{x})=\tau -1$. In other words, with non-zero probability, $g$ is equal to $\textsf{THR}_{L',\tau}$. 
Thus, $\textsf{THR}_{L',\tau}$ is a minor of~$f$. 
\end{proof}

Using the existence of arbitrarily large arity threshold minors, the algorithmic part of our Dichotomy result follows immediately.
\begin{theorem}
\label{thm:alg-main}
Let $\Gamma$ be a Promise CSP template. Suppose that for every $\epsilon >0$, there exists a function $f \in \mathrm{Pol}(\Gamma), f:\{0,1\}^n \rightarrow \{0,1\}$ such that $\si_i(f) \leq \epsilon$ for all $i \in [n]$. Then, $\mathrm{PCSP}(\Gamma)\in \emph{\textsf{P}}$.
\end{theorem}

\begin{proof}
Using~\Cref{lem:threshold-minor}, we can conclude that there are infinitely many positive integers~$L$ such that there exists $\tau \in \{0,1,\ldots, L\}$ with $\mathrm{\textsf{THR}}_{L,\tau}\in \mathrm{Pol}(\Gamma)$. As the threshold functions are symmetric, $\mathrm{Pol}(\Gamma)$ has symmetric polymorphisms of infinitely many arities. Thus, using the BLP+Affine algorithm of~\cite{BrakensiekGWZ20}, $\mathrm{PCSP}(\Gamma)$ can be solved in polynomial time. 
\end{proof}
We remark that the above result is inspired by a special case shown by Barto~\cite{Bar18} that a Boolean Ordered PCSP is polytime tractable if it has cyclic polymorphisms of arbitrarily large arities.

\section{Hardness Assuming Rich 2-to-1 Conjecture}
\label{sec:hardness}
In this section, we prove the hardness part of our dichotomy result. First, we prove that Shapley value is preserved under uniformly random $2$-to-$1$ minors, and then we use this to show the hardness assuming the Rich $2$-to-$1$ Conjecture. 
\subsection{Shapley value under random 2-to-1 minor}

Let $f: \{ 0,1 \}^{2n}\rightarrow \{0,1\}$ be a monotone Boolean function with $\si_f(1)\geq \lambda$ for some absolute constant $\lambda >0$. 
Let $g:\{0,1\}^n \rightarrow \{0,1\}$ be a minor of $f$ with respect to the uniformly random $2$-to-$1$ function $\pi : [2n] \rightarrow [n]$.
Our goal in this subsection is to show that $\E_{\pi}[\si_{g}(\pi(1))]\geq \gamma$ for some function $\gamma := \gamma(\lambda)>0$.
\input{two-step-figure}
We prove this in two steps. (See~\Cref{fig:two-step-minors})
\begin{enumerate}
    \item First, we consider the minor of $f$, $f' : \{0,1\}^{2n-1} \rightarrow \{0,1\}$ obtained with respect to $\pi_1 : [2n] \rightarrow [2n-1]$ where $\pi_1(1)=\pi_1(2)=1, \pi_1(i)=i-1\, \forall i \in \{3,4,\ldots, 2n\}$. We show that $\si_{f'}(1)\geq \frac{\lambda}{2}$. 
    \item Next, we consider a minor $g$ of $f'$ obtained with respect to the function $\pi_2 : [2n-1] \rightarrow [n]$ which has $\pi_2(1)=1$ while the remaining $2n-2$ values are chosen using a uniformly random partition of $[2n-2]$ into $n-1$ pairs. We show that $\E_{\pi_2}[\si_{g}(1)]\geq \gamma$ for some function $\gamma := \gamma(\lambda)>0$.
\end{enumerate}
The two steps together prove that when $g$ is a minor of $f$ with respect to the function $\pi:[2n]\rightarrow [n]$ that is a uniformly random $2$-to-$1$ function conditioned on the fact that $\pi(1)=\pi(2)$, we have $\E_{\pi}[\si_{g}(\pi(1))]\geq \gamma$ for some function $\gamma := \gamma(\lambda)>0$. Taking average over all $i$ such that $\pi(1)=\pi(i)$, we get the same claim when $g$ is a uniformly random $2$-to-$1$ minor.


The first step is captured by the following lemma.
\begin{lemma}
\label{lem:almost-same-minor}
Let $f:\{0,1\}^{2n} \rightarrow \{0,1\}$ and $f':\{0,1\}^{2n-1} \rightarrow \{0,1\}$ be monotone Boolean functions such that $f'$ is a minor of $f$ with respect to the function $\pi_1 : [2n] \rightarrow [2n-1]$ defined as $\pi_1(i)=\max(i-1,1)$. If $\si_f(1)\geq \lambda$, then  $\si_{f'}(1)\geq \frac{\lambda}{2}$.
\end{lemma}
\begin{proof}
We recall a bit of notation: let $\mathcal{B}_f(1)$ denote the boundary of the coordinate $1$ in the function $f$ i.e., the family of all the sets $S \subseteq [2n]\setminus \{1\}$ such that $f(S)=0, f(S\cup \{1\})=1$. For an integer $j \in \{0,1,\ldots,2n-1\}$, let $\mu_f(j)^{(1)}$ denote the fraction of subsets of $[2n]\setminus\{1\}$ of size $j$ that are in $\mathcal{B}_f(1)$.
For ease of notation, we let $\mu(j)=\mu_f(j)^{(1)}$, and $\mu'(j)=\mu_{f'}(j)^{(1)}$.
Consider a set $S \subseteq [2n]\setminus \{1\}$ such that $S \in \B_f(1)$. Note that 
\[
S' = \{ i-1 : i > 2, i \in S\}
\]
satisfies $S' \in \B_{f'}(1)$. Suppose that $S_1, S_2 \in \B_f(1)$ such that $|S_1|=|S_2|=j$, $S_1 \neq S_2$ and $2 \notin S_1 \cup S_2$. Then, the above definition satisfies $S'_1 \neq S'_2$, $S'_1, S'_2 \in \B_{f'}(1)$ and $|S'_1|=|S'_2|=j$. This implies that 
\[
\Big| \{ S : S \in \B_f(1), |S|=j, 2 \notin S\}\Big| \leq \left| \B_{f'}(1) \cap \binom{[2n-1]\setminus\{1\}}{j} \right|
\]
Similarly, 
\[
\Big| \{ S : S \in \B_f(1), |S|=j, 2 \in S\}\Big| \leq \left| \B_{f'}(1) \cap \binom{[2n-1]\setminus\{1\}}{j-1} \right|
\]
Summing the two, we obtain that 
\[
\Big| \{ S : S \in \B_f(1), |S|=j\}\Big| \leq \left| \B_{f'}(1) \cap \binom{[2n-1]\setminus\{1\}}{j-1} \right|
+ \left| \B_{f'}(1) \cap \binom{[2n-1]\setminus\{1\}}{j} \right|
\]
We can rewrite it as 
\[
    \binom{2n-1}{j}\mu(j)\leq \binom{2n-2}{j}\mu'(j)+\binom{2n-2}{j-1}\mu'(j-1) \, \forall j \in [2n-2]
\]

As $\binom{2n-1}{j}=\binom{2n-2}{j}+\binom{2n-2}{j-1}$ for every $j\in [2n-2]$, we get that 
\[
\mu(j) \leq \mu'(j) + \mu'(j-1)
\]
for all $j \in [2n-2]$.
Also note that $\mu(0)= \mu'(0)$, and $\mu(2n-1)= \mu'(2n-2)$. 
Summing over all these inequalities, we get that 
\[
\sum_{j \in \{0,1,\ldots, 2n-2\}} \mu'(j) \geq \frac{1}{2} \sum_{j \in \{0,1,\ldots,2n-1\}} \mu(j) \geq \frac{ \lambda (2n)}{2} = n \lambda
\]
Thus, 
\[
\si_{f'}(1) = \frac{\sum_{j \in \{0,1,\ldots,2n-2\}} \mu'(j)}{2n-1} \geq \frac{\lambda}{2} \  . \qedhere
\]
\end{proof}

Before proving the second step, we prove the following key lemma regarding the distribution of the boundary subsets. 
\begin{lemma}
\label{lem:even}
Let $f':\{0,1\}^{2n-1}\rightarrow \{0,1\}$ be a monotone Boolean function such that $\si_{f'}(1)=\lambda$ with $\lambda \geq \frac{1}{n}$. For an integer $j \in \{0,1,\ldots, 2n-2\}$, let $\mu'(j)=\mu_{f'}(j)^{(1)}$.
Then, there exists an absolute constant $\gamma := \gamma(\lambda)>0$ such that 
\[
\frac{\sum_{j=0}^{n-1} \mu'(2j)}{n}\geq \gamma
\]
\end{lemma}
\begin{proof}
We prove that there exist real numbers $c_1< c_2, c>\frac{\lambda^2}{4}$, such that for all $j$ with $c_1 n \leq j \leq c_2 n$, we have $\mu'(j)\geq c$, and $c_2 - c_1 \geq \frac{\lambda^2}{8}$. This directly implies the lemma with $\gamma = \Omega(\lambda^4)$.

For a pair of integers $0\leq i<j\leq 2n-2$, we define the following parameter $\mu'(i,j)$ as the fraction of the pair of subsets $(S,T)$ where $S,T \subseteq \{2,3,\ldots, 2n-1\}, |S|=i, |T|=j, S \subseteq T$ that satisfy $S \in \B_{f'}(1),T \in \B_{f'}(1)$. 
\[
\mu'(i,j) = \frac{\left| \{ (S,T): |S|=i, |T|=j, S \subseteq T, S \in \B_{f'}(1), T \in \B_{f'}(1)\}\right|}{\binom{2n-2}{i} \binom{2n-2-i}{j-i}}
\]

We first claim that there exist constants (depending on $\lambda$) $c_1< c_2,c>0$ such that $\mu'(c_1n,c_2n)\geq c$, and $c_2 - c_1 \geq \frac{\lambda^2}{2}$.
Consider a uniformly random permutation of $[2n-1]\setminus \{1\}$ denoted by $\sigma = (\sigma(1), \sigma(2), \ldots, \sigma(2n-2))$. For an integer $j \in \{0,1,\ldots, 2n-2\},$ let $S_j$ be the random variable that is the union of the prefix of $\sigma$ containing the first $j$ elements. 
\[
S_j := \{ \sigma(1), \sigma(2), \ldots, \sigma(j) \}, \, \forall j \in \{0,1,\ldots, 2n-2\}. 
\]
For each $j \in \{0,1,\ldots, 2n-2\}$, the subset $S_j$ is uniformly distributed in $\binom{[2n-1]\setminus\{1\}}{j}$. 
For $j \in \{0,1,\ldots, 2n-2\},$ let $X_j$ be the indicator random variable for the event that $S_j \in \mathcal{B}_{f'}(1)$. 
By the definition of $\mu'(j)$, we get
\[
\E[X_j] = \mu'(j) \, \ \forall j \in \{0,1,\ldots,2n-2\}.
\]
Let $X = X_0 + X_1 +\ldots+X_{2n-2}$ be the number of subsets in the set family
$(\phi =S_0 \subset S_1 \subset S_2 \ldots \subset S_{2n-2}=[2n-1]\setminus\{1\}) $ that are in $\mathcal{B}_{f'}(1)$. 
Using~\Cref{eq:si-alternate}, we get
\[
\E[X]=\lambda(2n-1).
\]
Using Jensen's inequality, we get that 
\[
\E\left[\binom{X}{2}\right] \geq \binom{\lambda(2n-1)}{2} = \frac{1}{2}\cdot \lambda(2n-1) \left( \frac{\lambda}{2}(2n-2)+n\lambda-1\right)\geq  \frac{\lambda^2}{2}\binom{2n-1}{2}
\]
wherein the final inequality, we used the fact that $\lambda n \geq 1$.
Note that for every $i <j$, the marginal distribution of $(S_i, S_j)$ is the uniform distribution over all the pairs of subsets $(S,T)$ where $S,T \subseteq \{2,3,\ldots, 2n-1\}, |S|=i, |T|=j, S \subseteq T$. 
Thus, by the definition of $\mu'(i,j)$, we get that 
    $\mu'(i,j) = \E [X_i X_j]$, for $0 \leq i < j \leq 2n-2$.
Therefore we have 
\begin{equation*}
\E \left[ \binom{X}{2}\right] = \E \left[\sum_{ 0 \leq i < j \leq 2n-2} X_i X_j \right] 
= \sum_{ 0 \leq i < j \leq 2n-2} \E [X_i X_j] = \sum_{ 0 \leq i < j \leq 2n-2} \mu'(i,j) 
\end{equation*}
Thus, 
\[
\sum_{ 0 \leq i < j \leq 2n-2} \mu'(i,j) \geq \frac{\lambda^2}{2}\binom{2n-1}{2}
\]
This implies that the expected value (over $i,j$) of $\mu'(i,j)$ is at least $\frac{\lambda^2}{2}$. Let $p$ denote the probability (over $i,j$) that $\mu'(i,j) \leq \frac{\lambda^2}{4}$. As the expected value of $\mu'(i,j)$ is at least $\frac{\lambda^2}{2}$, we have 
\[
p\cdot \left( \frac{\lambda^2}{4} \right) + (1-p) \cdot 1 \geq \frac{\lambda^2}{2}
\]
which implies that $p \leq 1-\frac{\lambda^2}{4}$.
Thus, with probability (over $i,j$) at least $\frac{\lambda^2}{4}$, we have $\mu'(i,j)\geq \frac{\lambda^2}{4}$.
Hence, there exist integers $p,q$ such that $q-p \geq \frac{\lambda^2}{8}n$ and $\mu'(p,q)\geq \frac{\lambda^2}{4}$. For ease of notation, let $p=c_1n, q = c_2n$ where $c_1,c_2$ are reals satisfying $c_2-c_1\geq \frac{\lambda^2}{8}$.

Next, we claim that $\mu'(j)\geq \frac{\lambda^2}{4}$ for all $j$ such that $c_1 n \leq j \leq c_2 n$. Fix an integer $j$ with $c_1 n \leq j \leq c_2 n$. Consider a uniformly random sequence of subsets $S_1 \subseteq S_2 \subseteq S_3 \subseteq [2n-1] \setminus \{1\}$ such that $|S_1|=c_1n, |S_2|=j, |S_3|=c_2n$. The probability that $S_1 \in \B_{f'}(1), S_3 \in \B_{f'}(1)$ is equal to $\mu'(c_1n, c_2n)$ which is at least $\frac{\lambda^2}{4}$. Thus, using~\Cref{prop:sandwich}, with probability at least $\frac{\lambda^2}{4}$, $S_2 \in \B_{f'}(1)$.
Note that the distribution of $S_2$ is uniform in $\binom{[2n-1]\setminus \{1\}}{j}$, and thus, we have $\mu'(j)\geq \frac{\lambda^2}{4}$. 

The fact that $\mu'(j)\geq \frac{\lambda^2}{4}$ for all $j$ such that $c_1 n \leq j \leq c_2 n$ together with $c_2-c_1 \geq \frac{\lambda^2}{8}$ completes a proof of the lemma. 
\end{proof}

We now prove the second step in the proof. 
\begin{lemma}
\label{lem:main}
Suppose that $f': \{0,1\}^{2n-1}$ is a monotone Boolean function such that $\si_{f'}(1)\geq \lambda$ with $\lambda \geq \frac{1}{n}$. Let $g$ be a random minor of $f'$ with respect to $\pi_2 : [2n-1] \rightarrow [n]$ which is obtained by setting $\pi_2(1)=1$, and for every $i >1$, we randomly choose $j_1, j_2 \in [2n-1]\setminus \{1\}$ (without replacements) and set $\pi_2(j_1)=\pi_2(j_2)=i$. In other words, we choose a uniformly random partition of $[2n-1]\setminus \{1\}$ into $n-1$ pairs $P_2, P_3, \ldots, P_{n}$ and set $\pi_2(j)=i\, \forall j \in P_i$. Then, there exists $\gamma := \gamma(\lambda)>0$ such that 
\[
\E_{\pi_2}[\si_g(1)]\geq \gamma \ .
\]
\end{lemma}
\begin{proof}
For ease of notation, we let $\mu'(j)=\mu_{f'}(j)^{(1)}$ and $\mu_g(j)=\mu_{g}(j)^{(1)}$.
For a set $S \subseteq [n]\setminus \{1\}$ and a function $\pi_2 : [2n-1] \rightarrow [n]$ with $\pi_2(1)=1$, and $|\pi_2^{-1}(i)|=2$ for all $i \in \{ 2,3,\ldots, n\}$, let $\pi_2^{-1}(S)$ be the $2|S|$ sized subset of $\{2,3,\ldots, 2n-1\}$ defined as follows: 
\[
\pi_2^{-1}(S):= \{ j \in \{2,3,\ldots, 2n-1\} :\pi_2(j) \in S\}
\]
For every set $S \subseteq \{2,3,\ldots, n\}$, when $\pi_2 :[2n-1] \rightarrow [n]$ is a uniformly random $2$-to-$1$ minor with $\pi_2(1)=1$, and the rest $2n-2$ elements are partitioned into $n-1$ pairs uniformly at random, the set $\pi_2^{-1}(S)$ is distributed uniformly in $\binom{[2n-1]\setminus \{1\}}{2|S|}$. 
Also note that $S \in \B_{g}(1)$ if and only if $\pi^{-1}(S) \in \B_{f'}(1)$. Thus, for every set $S \subseteq \{2,3,\ldots, n\}$, the probability that $S \in \B_g(1)$ (over the choice of $\pi_2$) is equal to $\mu'(2|S|)$. Summing over all such sets of size $j$, we get that for every $j \in \{0,1,\ldots,n-1\}$, the expected value of $\mu_g(j)$ is equal to $\mu'(2j)$.
\[
    \label{eq:main}
\mathbb{E}_{\pi_2}[\mu_g(j)]=\mu'(2j)\, \forall j \in \{0,1,\ldots, n-1\}
\]
By using~\Cref{lem:even}, we can infer that there exists $\gamma = \gamma(\lambda) >0$ such that $\sum_{j=0}^{n-1}\mathbb{E}_{\pi_2}[\mu_g(j)]=\sum_{j=0}^{n-1}\mu'(2j) \geq \gamma n$. 
Using~\Cref{eq:si-alternate}, we get 
\begin{align*}
    \E_{\pi_2}[\si_g(1)]  & = \mathbb{E}_{\pi_2}\left[\frac{\sum_{j=0}^{n-1} \mu_g(j)}{n}\right]
    = \frac{\sum_{j=0}^{n-1}\mathbb{E}_{\pi_2}[\mu_g(j)]}{n} \geq \gamma.  \qedhere
\end{align*}
\end{proof}

\Cref{lem:almost-same-minor} and~\Cref{lem:main} together prove that Shapley value behaves well under uniformly random $2$-to-$1$ minors for monotone Boolean functions. 

\begin{lemma}
\label{lem:random-minor}
Suppose that $f:\{0,1\}^{2n}\rightarrow \{0,1\}$ is a monotone Boolean function such that $\si_f(1)\geq \lambda$ for some absolute constant $\lambda >0$ with $\lambda \geq \frac{1}{n}$. Then, there exists $\gamma := \gamma(\lambda)>0$ such that 
\[
\E_{\pi}[\si_g(\pi(1))]\geq \gamma
\]
where $g$ is a minor of $f$ with respect to the uniformly random $2$-to-$1$ function $\pi$. 
\end{lemma}
\begin{proof}
Combining~\Cref{lem:almost-same-minor} and~\Cref{lem:main}, we can conclude that for every $i \in [2n], i >1$, when $\pi :[2n] \rightarrow [n]$ is a uniformly random $2$-to-$1$ minor conditioned on the fact that $\pi(1)=\pi(i)$, we have $\E_{\pi}[\si_g(\pi(1))]\geq \gamma$. 
Taking average over all the $i \in [2n], i>1$, we get a proof that the same inequality holds when $\pi$ is a uniformly random $2$-to-$1$ minor. 
\end{proof}
\subsection{Reduction}

We first formally define the Label Cover problem and state the Rich $2$-to-$1$ Conjecture. 

\begin{definition}(Label Cover)
In the Label Cover problem $\mathcal{G}=(G,\Sigma_L,\Sigma_R,\Pi)$, the input is a bipartite graph $G=(L \cup R, E)$ with projection constraint $\Pi_e :\Sigma_L \rightarrow \Sigma_R$ on every edge $e \in E$. A labeling $\sigma$ which assigns values from $\Sigma_L$ to $L$ and from $\Sigma_R$ to $R$ satisfies the constraint $\Pi_e$ on the edge $e=(u,v)$ if $\Pi_e(\sigma(u))=\sigma(v)$. The objective is to identify if there is a labeling that satisfies all the constraints. 
\end{definition}

For every constant $\epsilon >0$, it is NP-hard~\cite{Raz98} to distinguish between the case that a given Label Cover instance has a labeling that satisfies all the constraints vs.\,no labeling can satisfy more than $\epsilon$ fraction of the constraints. 
This hardness result for Label Cover has been instrumental in showing numerous strong, and sometimes optimal, inapproximability results for various computational problems. However, the standard Label Cover seems insufficient as a starting point towards proving hardness results for approximate graph coloring and other $2$-variable PCSPs. To circumvent this, the hardness of Label Cover on structured instances such as Unique Games, smooth Label Cover, etc. has been studied. 

In his celebrated work proposing the Unique Games Conjecture~\cite{Khot02}, 
Khot also proposed the ``$2$-to-$1$ conjecture'' that the strong hardness of Label Cover holds when all the constraints of the Label Cover are $2$-to-$1$ functions. The imperfect completeness version of this conjecture was recently established in a striking sequence of works~\cite{KMS17,DKKMS18a,DKKMS18,KMS18}. Braverman, Khot, and Minzer~\cite{BravermanKM19} put forth a stronger conjecture that states that the hardness of Label Cover holds when the distribution of $2$-to-$1$ functions on edges incident on every vertex $u \in L$ is the uniform distribution. For ease of notation, for an integer $n$, we use $\mathcal{F}_{2\rightarrow 1}(n)$ to denote the set of all the $2$-to-$1$ functions from $[2n]$ to $[n]$. 

\begin{definition}(Rich $2$-to-$1$ Label Cover instances)
We call a Label Cover instance $\mathcal{G}=(G,\Sigma_L,\Sigma_R,\Pi)$ with $G=(L \cup R,E)$ a rich $2$-to-$1$ instance if the following hold.
\begin{enumerate}
    \item There exists an integer $\Sigma$ such that $\Sigma_L=[2\Sigma]$, $\Sigma_R=[\Sigma]$, and every projection constraint $\Pi_e$, $e \in E$ is a $2$-to-$1$ function. 
    \item For every vertex $u \in L$, the distribution of $2$-to-$1$ functions $\mathcal{P}_u$ obtained by first sampling a uniformly random neighbor $v$ of $u$, and then picking $\Pi_e, e=(u,v)$, is uniform over $\mathcal{F}_{2\rightarrow 1}(\Sigma)$.
\end{enumerate}
\end{definition}

\begin{conjecture}(Rich $2$-to-$1$ Conjecture)~\cite{BravermanKM19} For every $\epsilon >0$, there exists an integer $\Sigma = \Sigma(\epsilon)$ such that given a rich $2$-to-$1$ Label Cover instance $\mathcal{G}=(G,\Sigma_L,\Sigma_R,\Pi)$ with $\Sigma_L=[2\Sigma]$, it is NP-Hard to distinguish between the following. 
\begin{enumerate}
    \item There is a labeling that satisfies all the constraints of $\mathcal{G}$. 
    \item No labeling can satisfy more than $\epsilon$ fraction of the constraints of $\mathcal{G}$.
\end{enumerate}
\end{conjecture}

We are now ready to state the hardness part of our dichotomy. It is proved using the Label Cover-Long Code framework.  
This reduction is standard in the PCSP literature, see e.g., ~\cite{barto2019algebraic}. 

\begin{theorem}
\label{thm:hardness}
Assume the Rich $2$-to-$1$ Conjecture. Let PCSP$(\Gamma)$ be a Boolean Ordered PCSP such that there exists an absolute constant $\lambda>0$ with $\max_{i \in [n]}\si_f(i) \geq \lambda$
for all functions $f:\{ 0,1 \}^n \rightarrow \{0,1\}$, $f \in \emph{Pol}(\Gamma)$. Then PCSP$(\Gamma)$ is NP-Hard. 
\end{theorem}
\begin{proof}
Let $\Gamma = \{ (A_1, B_1), (A_2, B_2), \ldots, (A_l, B_l) \}$ be the PCSP under consideration, where each~$A_i$ is a subset of $\{0,1\}^{k_i}$ for all $i \in [l]$, and similarly, each $B_i$ is a subset of $\{0,1\}^{k_i}$ for all $i \in [l]$. 
We start from a rich $2$-to-$1$ Label Cover instance $\mathcal{G}=(G,[2\Sigma],[\Sigma],\Pi)$ with $G=(L \cup R, E)$. 
For ease of notation, we use $\Sigma_w$ to denote $2\Sigma$ if $w \in L$, and $\Sigma$ if $w \in R$.
For every vertex $w \in L \cup R$, we have a set of $2^{\Sigma_w}$ nodes denoted by $L_w = \{ w \} \times \{0,1\}^{\Sigma_w}$ referred to as the long code corresponding to $w$. 
The elements of our output PCSP instance $V$ is the union of all the long code nodes. 
\[
V = \bigcup_{ w \in L \cup R}L_w
\]
We add two types of constraints. 
\begin{enumerate}
    \item Polymorphism Constraints. For every $i \in [l]$, we add the following constraints using the pair of predicates $(A_i, B_i)$. 
    For every $w \in L\cup R$, and multiset of vectors $\textbf{x}^1, \textbf{x}^2, \ldots, \textbf{x}^{k_i} \in \{0,1\}^{\Sigma_w}$ satisfying 
    \[
        ( \textbf{x}^1_j, \textbf{x}^2_j, \ldots, \textbf{x}^{k_i}_j) \in A_i \, \forall j \in [\Sigma_w], 
    \]
    we add the constraint on the $k_i$ nodes $\{ w, \textbf{x}^1\},\{w,\textbf{x}^2\}, \ldots, \{w,\textbf{x}^{k_i}\}$.
    \item Equality Constraints. For every edge $e=(u,v)$ of the Label Cover instance with the constraint $\Pi_e : [2\Sigma]\rightarrow [\Sigma]$, we add the following set of equality constraints. For every $\textbf{x} \in \{0,1\}^{2\Sigma}$ and $\textbf{y} \in \{0,1\}^\Sigma$ such that for all $j \in [2\Sigma]$, $\textbf{x}_{j}=\textbf{y}_{\Pi_e(j)}$, we add an equality constraint between $\{u,\textbf{x}\}$ and $\{v,\textbf{y}\}$ ensuring that the two nodes are assigned the same value. The fact that we can add the equality constraints follows either by identifying the variables together, or by observing that the polymorphism minion of any PCSP remains the same when we add the equality predicate (see e.g.,~\cite{barto2019algebraic,GuruswamiS20}). 
\end{enumerate}

\paragraph{Completeness.}
Suppose that there exists a labeling $\sigma$ that satisfies all the constraints of the Label Cover instance. 
For every node $\{w,\textbf{x}\} \in V$, we assign the dictator function $\textbf{x}_{\sigma(w)} \in \{0,1\}$. By the way we have added the polymorphism constraints, any dictator assignment satisfies them. The equality constraints are also satisfied as the labeling satisfies all the constraints of $\mathcal{G}$.  

\paragraph{Soundness.}
Suppose that there exists an assignment $f: V \rightarrow \{0,1\}$ that satisfies all the polymorphism constraints and the equality constraints. Then, we claim that there exists a labeling $\sigma$ that satisfies $\epsilon := \epsilon(\lambda)>0$ fraction of the constraints of the Label Cover instance $\mathcal{G}$. 

For a vertex $w \in L \cup R$, let $f_w : \{0,1\}^{\Sigma_w} \rightarrow \{0,1\}$ denote the function $f$ restricted to $L_w$. Note that $f_w$ is a polymorphism of the PCSP $\Gamma$ for all $w \in L \cup R$. As every polymorphism of $\Gamma$ has a coordinate with Shapley value at least $\lambda$, for every $u \in L$, we define the set $S(u)$ that is non-empty as follows:
\[
S(u) = \{ i \in [2\Sigma] : \si_{f_u}(i) \geq \lambda \}
\]
As $\sum_{i \in [n]}\si_f(i)=1$ for all functions $f:\{0,1\}^n \rightarrow \{0,1\}$, we have $|S(u)|\leq \frac{1}{\lambda}$ for all $u \in L$. 

As a corollary of~\Cref{lem:random-minor}, we can conclude that there exists $\gamma  = \gamma(\lambda)>0$ such that for every monotone Boolean function $f: \{0,1\}^{2\Sigma} \rightarrow \{0,1\}$ with $\si_f(i)\geq \lambda$, when $g$ is a minor of $f$ with respect to a uniformly random $2$-to-$1$ function $\pi : [2\Sigma]\rightarrow [\Sigma]$, $\si_g(\pi(1)) \geq \frac{\gamma}{2}$ with probability at least $\frac{\gamma}{2}$.
Note that applying~\Cref{lem:random-minor} requires that $\lambda \geq \frac{1}{\Sigma}$. However, even when $\lambda < \frac{1}{\Sigma}$, by picking the coordinate with the largest Shapley value, we can still assume that in every long code function, there is a coordinate with Shapley value at least $\frac{1}{2\Sigma}=\Theta(\lambda)$, and then apply~\Cref{lem:random-minor}.  
Using this $\gamma$, for every $v \in R$, we define the set $S(v)$ as
\[
S(v) = \left\{ i \in [\Sigma] : \si_{f_v}(i) \geq \frac{\gamma}{2} \right\} .
\]
By definition, we have $|S(v)|\leq \frac{2}{\gamma}$ for all $v \in R$. 
As the Label Cover instance is rich $2$-to-$1$, for every $u \in L$, when we pick a uniformly random edge $e=(u,v)$ adjacent to $u$ with constraint $\Pi_e : [2\Sigma]\rightarrow[\Sigma]$, with probability at least $\frac{\gamma}{2}$,  there exist $i_1 \in [2\Sigma], i_2 \in [\Sigma]$ such that $\si_{f_u}(i_1) \geq \lambda$, $\si_{f_v}(i_2) \geq \frac{\gamma}{2}$, and $\Pi_e(i_1)=i_2$. 

We now pick a labeling $\sigma$ of $\mathcal{G}$ by picking uniformly random label from $S(w)$ for all $w \in L \cup R$.
By the above argument, for every $u \in L$, the expected number of constraints of $\mathcal{G}$ that are adjacent to $u$ that the labeling $\sigma$ satisfies is at least $\frac{\gamma}{2} \cdot \lambda \frac{\gamma}{2}$. Summing over all $u \in L$, $\sigma$ satisfies at least $\Omega(\lambda \gamma^2)$ fraction of the constraints of $\mathcal{G}$ in expectation. 
Thus, there exists a labeling to $\mathcal{G}$ that satisfies $\epsilon = \Omega(\lambda \gamma^2)>0$ fraction of the constraints, which completes the proof. \qedhere 
\end{proof}

\section{Adversarial 2-to-1 minor}
\label{sec:adversarial}
We construct an example of a $2$-to-$1$ minor where the Shapley value alters completely after taking the minor. 
\begin{theorem}
\label{thm:example}
Let $n \geq 2$ be a positive integer. There exist two monotone Boolean functions $f:\bits^{2n}\rightarrow \bits$ and $g:\bits^{n}\rightarrow\bits$ such that $g$ is a $2$-to-$1$ minor of $f$ with respect to the $2$-to-$1$ function $\pi : [2n]\rightarrow [n]$ defined as $\pi(i)=\lceil \frac{i}{2}\rceil$. Furthermore, 
\begin{enumerate}
    \item $\si_g(1)=\Omega(1)$, and $\si_g(j)=o(1)$ for all $j >1$. 
    \item $\si_f(3)=\Omega(1)$, and $\si_f(i)=o(1)$ for all $i \in [2n], i \neq 3$.
\end{enumerate}
\end{theorem}
\begin{proof}
Similar to the proof of~\Cref{thm:hardness}, we construct the minor function pair in two steps. 
\begin{enumerate}
    \item First, we construct Boolean monotone functions $f':\{0,1\}^{2n-1} \rightarrow \{0,1\}$ and $g:\{0,1\}^{n}\rightarrow \{0,1\}$ such that $g$ is a minor of $f$ with respect to the function $\pi : [2n-1] \rightarrow [n]$ defined as $\pi(i)=\lceil \frac{i+1}{2} \rceil$ for all $i$. Furthermore, $\si_g(1)=\Omega(1)$, and $\si_g(j)=o(1)$ for all $j >1$. We also have $\si_{f'}(2)=\Omega(1)$, and $\si_{f'}(i)=o(1)$ for all $i \in [2n-1], i \neq 2$.
    \item We define the function $f:\{0,1\}^{2n}\rightarrow \{0,1\}$ as 
    \[
    f(y_1, y_2, \ldots, y_{2n})= f'(y_1, y_3, \ldots, y_{2n})
    \]
    Note that $g$ is a minor of $f$ with respect to the $2$-to-$1$ function $\pi :[2n]\rightarrow [n]$ defined as $\pi(i)=\lceil \frac{i}{2} \rceil$. Furthermore, by definition, we have $\si_{f}(3)=\Omega(1)$, and $\si_{f}(i)=o(1)$ for all $i \in [2n], i \neq 3$.
\end{enumerate}
Henceforth, our goal is to construct a pair of functions as in the first step above. 

We define a partial Boolean function to be a function from $\bits^n$ to $\{0,1,?\}$. 
    A partial Boolean function $f$ on $n$ variables is monotone if for all $\textbf{p} \in \{0,1\}^n$ and $\textbf{q} \in \{0,1\}^n$ such that $\textbf{p}\leq \textbf{q}$, if $f(\textbf{p})=1$, then $f(\textbf{q})=1$, and if $f(\textbf{q})=0$, then $f(\textbf{p})=0$. 
    
    Now, consider $g:\{0,1\}^{n}\rightarrow \bits$ to be 
    \[
    g(\textbf{x}) = \begin{cases}
    &1 \text{ if } \sum_{j=2}^n x_j \geq \frac{51n}{100} \\ 
    &0 \text{ if } \sum_{j=2}^n x_j \leq \frac{49n}{100} \\ 
    &x_1 \text{ if } \frac{49n}{100} < \sum_{j=2}^n x_j < \frac{51n}{100}
    \end{cases}
    \]
    
    By definition, $g$ is a monotone function, and using~\Cref{eq:si-alternate}, we can infer that $\si_g(1)=\frac{1}{50}$, and $\si_g(j) < \frac{1}{n}$ for all $j >1$. 

    We now construct $f'$ in three steps. 
    Start with $f'='?'$. 
    \begin{enumerate}
        \item (Preserving the minor) First, set the value of entries of $f'$ that are of the form\\ $(x_1,x_2,x_2,\cdots,x_n,x_n)$ as 
        \[
        f'(x_1, x_2, x_2, \ldots, x_n, x_n) = g(x_1, x_2, \ldots, x_n) \, \, \forall \textbf{x} \in \{0,1\}^n
        \]
        We then extend it both upwards and downwards i.e., if $f'(\textbf{p})$ is set to $1$ and $\textbf{p} \leq \textbf{q}$, then set $f'(\textbf{q})=1$ as well, and similarly, if $f'(\textbf{q})$ is set to $0$, and $\textbf{p} \leq \textbf{q}$, then we set $f'(\textbf{p})=0$. 
        This ensures that $g$ is a minor of $f'$ and that the partial function $f'$ is monotone.

        \item (Destroying the influence of $1$) Next, we ensure that the Shapley value of the coordinate $1$ is low by the following operation: 
        consider all $\textbf{y}$ such that $f'(\textbf{y})$ has not been set in the first step, $y_1=0$ and $f'(1,y_2,\cdots,y_{2n-1})$ is already set to $1$ in the first step. Then set $f'(\textbf{y})$ to be $1$. 
        Similarly, if $\textbf{y}$ satisfies $y_1=1$ and $f'(0,y_2,\cdots,y_{2n-1})$ is already set to $0$ in the first step, set $f'(\textbf{y})$ to be $0$ if it has not been set in the first step.

        We claim that the updated partial function $f'$ is still a monotone partial function. 
        Consider $\textbf{p}, \textbf{q} \in \{0,1\}^{2n-1}$ such that $\textbf{p} \leq \textbf{q}$.
        Suppose that $f'(\textbf{p})$ is set to be $1$. If it is set in the first step, as we extended the partial function upwards in the first step, $f'(\textbf{q})=1$ as well. If $f'(\textbf{p})$ is set to be $1$ in the second step, it implies that $f'(\textbf{p}')$ has been set to $1$ in the first step, where $\textbf{p}'$ is obtained from $\textbf{p}$ by setting $p_1$ to be $1$. 
        Let $\textbf{q}' \in \{0,1\}^{2n-1}$ be obtained from~$\textbf{q}$ by setting $q_1=1$. As $\textbf{p}' \leq \textbf{q}'$, $f'(\textbf{q}')$ has been set to $1$ in the first step as well. Thus, $f'(\textbf{q})$ is set to be $1$ in the second step. 
        The same argument can be used to show that if $f'(\textbf{q})=0$, then $f'(\textbf{p})=0$ as well. 
    \item (Adding influence to $2$) For all $\textbf{y}$ for which $f'(\textbf{y})='?'$ set $f'(\textbf{y})=y_2$. The fact that the final function $f'$ is monotone follows from observing that any completion of a partial monotone function using a monotone function results in a monotone function. 
    \end{enumerate}

Finally, our goal is to argue about the Shapley value of the coordinates of the function~$f'$. 
First, we show that the Shapley value of the coordinate $1$ in $f'$ is $o(1)$.
Suppose there exists $\textbf{p}=(0,y_2,y_3,\cdots,y_{2n-1})$ and $\textbf{q}=(1,y_2,y_3,\cdots,y_{2n-1})$ such that $f'(\textbf{p})=0$ and $f'(\textbf{q})=1$. 
We claim that both the values $f'(\textbf{p})$ and $f'(\textbf{q})$ are set in the first step of the above procedure. 
Suppose for contradiction that this is not the case. If neither of them is set in the first step, then they will not be set in the second step either, and in the third step, both of them will be assigned the same value, a contradiction. 
If exactly one of them is set in the first step, then in the second step, the other value would be set to be equal to it, a contradiction as well. 
Thus, both the values $f'(\textbf{p})$ and $f'(\textbf{q})$ are set in the first step. 
    
    Let $B =\mathcal{B}_g(1) \subseteq \{0,1\}^{n-1}$ be the boundary of the coordinate $1$ in $g$. 
    As $f'(\textbf{q})$ is set to be $1$ in the first step, there exists $\textbf{x} \in \{0,1\}^n$ such that $g(\textbf{x})=1$ and $(x_1,x_2,x_2,\cdots,x_n,x_n) \leq \textbf{q}$. 
    As $(x_1,x_2,x_2,\cdots,x_n,x_n)$ is not less than or equal to $\textbf{p}$, we can conclude that $x_1=1$ and $g(0,x_2,x_3,\cdots,x_n)=0$. In other words, $(x_2,x_3,\cdots,x_n)\in B$. 
    Similarly, there exists $\textbf{x}'$ such that $g(\textbf{x}')=0$ and $(x'_1,x'_2,x'_2,\cdots,x'_n,x'_n) \geq \textbf{p}$. 
    By the same argument as above, we can conclude that $(x'_2,x'_3,\cdots,x'_n)\in B$. Combining the both, we can conclude that there exist $\textbf{x},\textbf{x}' \in B$ such that $(x_2, x_2, x_3, x_3, \ldots, x_n,x_n) \leq (y_2,y_3,\cdots,y_{2n-2})\leq (x'_2, x'_2, x'_3, x'_3, \ldots, x'_n, x'_n)$. 
    Note that if the above inequality is true for a $(y_2,y_3,\cdots,y_{2n-2})$, we directly get that $(y_2,y_3,\cdots,y_{2n-2})$ is in the boundary of the coordinate $1$ in $f'$. 
    
    Observe that the boundary of coordinate $1$ in $g$ is the set of vectors $(x_2,x_3,\cdots,x_n)$ such that  $\frac{49}{100}n \leq \sum_{j=2}^n x_j \leq \frac{51}{100}n$.
  By the previous argument, we can deduce that the boundary $B'=\mathcal{B}_f(1)$ of the coordinate $1$ in $f'$ is the set of vectors $\textbf{y}=(y_2,y_3,\cdots,y_{2n-1})$ that satisfy the 
 following property: The number of $i\in [n-1]$ such that both $y_{2i}=y_{2i+1}=1$ is at least $\frac{49}{100}n$. Similarly, the number of $i\in [n-1]$ such that $y_{2i}=y_{2i+1}=0$ is at least
        $\frac{49}{100}n$. 
        Observe that this implies that we require that $\frac{49}{50}n \leq \sum_{j=2}^{2n-1}y_j \leq \frac{51}{50}n$.
    However, for every integer $l$ such that $\frac{49}{50}n \leq l \leq \frac{51}{50}n$, when we sample a uniformly random vector $\textbf{y}=(y_2, y_3, \ldots, y_{2n-1})$ with $\sum_{j=2}^{2n-1}y_j=l$, the probability that the number of $i\in [n-1]$ such that both $y_{2i}=y_{2i+1}=1$ is at least $\frac{49}{100}n$ is $o(1)$.
    Thus, using~\Cref{eq:si-alternate}, we can infer that the Shapley value of the coordinate $1$ in $f'$ is $o(1)$.

We now show that the coordinate $2$ has $\Omega(1)$ Shapley value in $f'$. Consider $\textbf{y}=(y_1, y_2, \ldots,\allowbreak y_{2n-1})$ such that $\frac{49n}{50}<\textsf{hw}(\textbf{y})\leq \frac{51n}{50}$. If the number of $i$ such that both $y_{2i}=y_{2i+1}=1$ is less than $\frac{49}{100}n$, we have $(y_1, y_3, \ldots, y_{2n-1}) \in \mathcal{B}_f(2)$. However, for every integer $l$ such that $\frac{49}{50}n \leq l \leq \frac{51}{50}n$, when we sample a uniformly random $\textbf{y}$ with $\textsf{hw}(\textbf{y})=l$, with probability $1-o(1)$, the number of $i$ such that both $y_{2i}=y_{2i+1}=1$ is less than $\frac{49}{100}n$. 
Thus, using~\Cref{eq:si-alternate}, we can infer that the Shapley value of the coordinate $2$ is $\Omega(1)$ in the function $f'$. 
Finally, by symmetry, we can observe that $\si_{f'}(i)=\si_{f'}(3)$ for all $i \geq 3$, and thus, as $\sum_i \si_{f'}(i)=1$, $\si_{f'}(i)=o(1)$ for all $i \geq 3$. 
\end{proof}

\section*{Acknowledgments}
We thank Libor Barto, whose talk~\cite{Bar18} and insightful discussions inspired our work.

\printbibliography

\end{document}

%% file: two-step-figure.tex
\begin{figure}
\centering
\begin{tikzpicture}[scale=0.4]
\draw (0,6)  node[fill,circle, draw, fill=black,inner sep=0,minimum size=0.15cm,label={\small $1$}](A) {};
\draw (2,6)  node[fill,circle, draw, fill=black,inner sep=0,minimum size=0.15cm,label={\small $2$}](B) {};
\draw (4,6)  node[fill,circle, draw, fill=black,inner sep=0,minimum size=0.15cm,label={\small $3$}](C) {};
\draw (6,6)  node[fill,circle, draw, fill=black,inner sep=0,minimum size=0.15cm,label={\small $4$}](D) {};
\draw (8,6)  node[fill,circle, draw, fill=black,inner sep=0,minimum size=0.15cm,label={\small $5$}](E) {};
\draw (10,6)  node[fill,circle, draw, fill=black,inner sep=0,minimum size=0.15cm,label={\small $6$}](F) {};

\node[] at (-2,6) {$f$};

\draw (1,3)  node[fill,circle, draw, fill=black,inner sep=0,minimum size=0.15cm,label={180:\small $1$}](G) {};
\draw (4,3)  node[fill,circle, draw, fill=black,inner sep=0,minimum size=0.15cm,label={180:\small $2$}](H) {};
\draw (6,3)  node[fill,circle, draw, fill=black,inner sep=0,minimum size=0.15cm,label={180:\small $3$}] (I){};
\draw (8,3)  node[fill,circle, draw, fill=black,inner sep=0,minimum size=0.15cm,label={180:\small $4$}] (J){};
\draw (10,3)  node[fill,circle, draw, fill=black,inner sep=0,minimum size=0.15cm,label={180:\small $5$}] (K){};

\node[] at (-2,3) {$f'$};

\draw (1,0) node[fill,circle, draw, fill=black,inner sep=0,minimum size=0.15cm,label={180:\small $1$}] (L){};
\draw (5,0) node[fill,circle, draw, fill=black,inner sep=0,minimum size=0.15cm,label={180:\small $2$}] (M){};
\draw (9,0) node[fill,circle, draw, fill=black,inner sep=0,minimum size=0.15cm,label={180:\small $3$}] (N){};

\node[] at (-2,0) {$g$};

\draw [->] (A) -- (G);
\draw [->] (B) -- (G);
\draw [->] (C) -- (H);
\draw [->] (D) -- (I);
\draw [->] (E) -- (J);
\draw [->] (F) -- (K);

\draw [->] (G) -- (L);
\draw [->] (H) -- (M);
\draw [->] (I) -- (M);
\draw [->] (J) -- (N);
\draw [->] (K) -- (N);

\end{tikzpicture}

\caption{An illustration of the two step minor approach: Here $f:\{0,1\}^6 \rightarrow \{0,1\}$ is a Boolean function, $f':\{0,1\}^5 \rightarrow \{0,1\}$ is a minor of $f$ with respect to the function $\pi_1 : [6] \rightarrow [5]$ with $\pi_1(i)=\max(i-1,1)$, and~$g$ is a minor of $f'$ with respect to the function $\pi_2 :[5]\rightarrow[3]$ with $\pi_2(i)=\lceil \frac{i+1}{2}\rceil$.}
\label{fig:two-step-minors}
\end{figure}
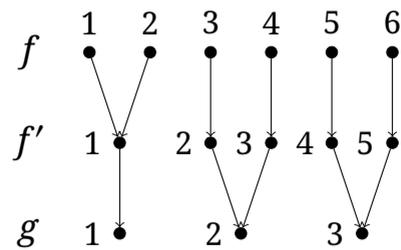